\documentclass[11pt,twoside]{article}

\usepackage[ansinew]{inputenc} 
\usepackage{amsmath,amsfonts,amssymb,amsthm}
\usepackage{bbm}

\newcommand{\field}[1]{\mathbb{#1}}
\newcommand{\N}{\field{N}}
\newcommand{\R}{\field{R}}
\newcommand{\C}{\field{C}}

\newcommand{\EE}{\mathcal E}

\newcommand{\eps}{\varepsilon}
\newcommand{\ph}{\varphi}
\newcommand{\const}{\mathrm{const}}

\newcommand{\norm}[1]{\mbox{$\left\| #1 \right\|$}}           
\newcommand{\sprod}[2]{\mbox{$\left\langle #1,#2 \right\rangle$}}        
\newcommand{\form}[3]{\mbox{$\left\langle #1\left|#2 \right|#3 \right\rangle $}}   

\newcommand{\supp}{\operatorname{supp}}

\newcommand{\curl}{\operatorname{curl}}

\newcommand{\dist}{\operatorname{dist}}

\newtheorem{theorem}{Theorem}[section]
\newtheorem{lemma}[theorem]{Lemma}
\newtheorem{corollary}[theorem]{Corollary}

\theoremstyle{plain}

\title{On the Magnetic Pekar Functional and the Existence of Bipolarons}

\author{
M.~Griesemer, F.~Hantsch and D.~Wellig\\
Universit\"at Stuttgart, Fachbereich Mathematik\\
70550 Stuttgart, Germany}

\date{}
\begin{document}
\maketitle
\begin{abstract}
First, this paper proves the existence of a minimizer for the
Pekar functional including a constant magnetic field and 
possibly some additional local fields that are energy reducing.
Second, the existence of the aforementioned minimizer is used to establish the binding of polarons in the model of 
Pekar-Tomasevich including external fields.
\end{abstract}

\section{Introduction}
The Pekar functional including external electric and magnetic potentials is given by
\begin{equation}\label{pekar}
  \int \Big(|D_A\ph|^2+V|\ph|^2\Big)\, dx - \int\frac{|\ph(x)|^2|\ph(y)|^2}{|x-y|}dxdy
\end{equation}
where $D_A:= -i\nabla+A$ and $\ph\in H^1_A(\R^3)$. The letters $V$ and $A$ denote (real-valued) scalar and vector potentials
associated with the external electric and magnetic fields $-\nabla V$ and $\curl A$. Since $\ph$ denotes
the wave function of a quantum particle (electron) we impose the constraint that
\begin{equation}\label{norm}
     \int|\ph|^2\, dx = 1.
\end{equation}
The functional \eqref{pekar} arises e.g. in the study of the ground
state energy of the polaron \cite{DV1983, LT1997} and in the analysis of a self-gravitating quantum particle \cite{MPT1997}.
Depending on the context, the Euler-Lagrange equation associated with \eqref{pekar}, \eqref{norm} is called Choquard equation or 
Schr\"odinger-Newton equation. The time-dependent version of the Euler-Lagrange equation describes the
dynamics of interacting many-boson systems in the mean field limit
\cite{EY2001}. We are interested in the question whether the
functional \eqref{pekar} subject to \eqref{norm}  has a minimizer,
and we shall give a positive answer for a class of potentials
including all previously considered cases. Second, we shall use the
existence of a minimizer to prove binding of polarons in the model
of Pekar and Tomasevich with an external magnetic field.

In the case $A=0$ and $V=0$ it is a well-known result, due to Lieb
\cite{Lieb77}, that the Pekar functional \eqref{pekar}, \eqref{norm} possesses a
unique, rotationally symmetric minimizer, which moreover can be
chosen pointwise positive. For the existence part a second proof
has been given by Lions as an application of his 
concentration compactness principle \cite{Lions1984}. Lions also considered the
case of non-vanishing $V\leq 0$. In this paper we establish
existence of a minimizer for constant magnetic fields and vanishing $V$, as well
as for certain local perturbations of this field configuration. For example, if $\curl A$ is constant, $V(x)=-|x|^{-1}$, then \eqref{pekar} has a minimizer as well.
More generally, the Pekar functional has a minimizer for any local perturbation
of  the fields $A(x)=(B\wedge x)/2$, $V=0$ that leads to a reduction of the energy. 
We give examples of non-linear vector potentials for which this trapping assumption is satisfied.

In the second part of the paper we address the question of binding of two polarons subject to given electromagnetic 
fields $A,V$ in the model of Pekar and Tomasevich. For  $A=0,V=0$ this question has been studied by Miyao, Spohn and by Lewin and answered 
in the affirmative for admissible values of the electron-electron repulsion close to the critical one  \cite{MS2007, Lewin2011}. 
In fact, Lewin proved the binding of any given number 
of polarons by establishing a Van der Waals type interaction between two polaron clusters. This method makes use of a spherical invariance 
which is broken by the presence of a magnetic field. We here describe a much softer argument to explain the binding of two polarons that works for 
any given $A,V$ and requires nothing but the existence of a minimizer  for \eqref{pekar}, \eqref{norm}. This argument is based on the observation that the product 
$\psi\otimes \psi$ of two copies of a minimizer $\psi$ of \eqref{pekar}, \eqref{norm} does not solve the Euler-Lagrange
equation of the Pekar-Tomasevich functional and hence cannot be a minimizer of this functional. This argument does not depend 
on the presence of external fields and seems to be novel. It can be extended to
multipolaron systems, and this will be done in subsequent work.

In a companion paper we derive estimates on the ground state energy of the Fr\"ohlich 
polaron subject to electromagnetic fields $A,V$ in the limit of strong electron-phonon coupling, $\alpha\to\infty$.
For fields  $A,V$ that are suitably rescaled with $\alpha$, it turns out that this ground state energy is correctly given by $\alpha^2$ times the minimum of 
\eqref{pekar}, \eqref{norm} up to errors of smaller order. In view of the results of the present paper the binding of Fr\"ohlich polarons subject to 
strong external fields and large $\alpha$ will follow. In the case $A=0$, $V=0$ a similar result 
has previously been established by Miyao and Spohn on the bases of \cite{DV1983, LT1997, Lieb77}. In the physical literature the 
existence of Fr\"ohlich bipolarons in the presence of magnetic fields is studied e.g. in \cite{BroDev1996}.

Solutions to the Choquard equation with magnetic field have very recently been studied in 
\cite{CSS2010, CSS2011}. In \cite{CSS2011} infinitely many solutions are found whose symmetry corresponds to the 
symmetry of $A$. Constant magnetic fields seem to be excluded, however. The constrained minimization problem \eqref{pekar}, \eqref{norm}
with non-vanishing magnetic field does not seem to have been studied yet. 
Nevertheless, as our methods are not new, we would not be surprised if some of our results on the existence of a minimizer 
for \eqref{pekar},\eqref{norm} with $A\neq 0$ could be inferred from existing results in the literature.  

Section 2 is devoted to the problem of existence of minimizers for \eqref{pekar}, \eqref{norm}; in Section~3 the binding 
of polarons is established. There is an appendix where technical auxiliaries are collected.\\ 

\emph{Acknowledgments.} Fabian Hantsch is supported by the \emph{Studienstiftung des Deutschen Volkes}, David Wellig has been supported by 
a stipend of the \emph{Landesgraduiertenf\"orderung} of Baden-W\"urttemberg.

\section{The Magnetic Pekar Functional}
This section contains all our results on the
existence of a minimizer for the Pekar functional, as well as the main parts of the proofs.
Some technical auxiliaries have been deferred to the appendix.

The \emph{minimal assumptions} that we shall make throughout the
paper, are that $A,V$ are real-valued with $A_k,V\in
L^2_{loc}(\R^3)$ and that $V$ is infinitesimally small with
respect to $-\Delta$, $V\ll -\Delta$. This means that for every
$\eps >0$ there exists $C_{\eps}\in\R$ such that
$$
\|V\ph\|\le\eps\|\Delta\ph\| + C_{\eps}\|\ph\|
$$
for all $\ph\in C_0^{\infty}(\R^3)$. Here and henceforth $\|\cdot\|$ denotes an $L^2$-norm.
Every potential $V$ that
admits a decomposition $V=V_1+V_2$ with $V_1\in L^2(\R^3)$ and
$V_2\in L^{\infty}(\R^3)$ is infinitesimally small w.r.t.
$-\Delta$.

We define $D_A := -i\nabla + A$ and
$$
   H^1_A(\R^3)=\big\{\ph\in L^2(\R^3)\mid D_A\ph \in L^2(\R^3;\C^3)\big\}.
$$
Equipped with the norm $\|\ph\|_A^2:=\|D_A\ph\|^2+\|\ph\|^2$ this
space is complete and $C_0^{\infty}(\R^3)$ is dense. This means
that the quadratic form $\sprod{D_A\ph}{D_A\ph}$ is closed on
$H^1_A(\R^3)$ and that $C_0^{\infty}(\R^3)$ is a core. The unique
self-adjoint operator associated with this form is denoted $D_A^2$.

We define the Pekar functional $\EE^{A,V}(\ph)$ by the expression
\eqref{pekar}. For the domain of this functional we take
$\left\{\ph\in H^1_A(\R^3)|\int |\ph|^2dx =1\right\}$ unless
explicitly stated otherwise. In particular, by a \emph{minimizer}
of $\EE^{A,V}$ we mean a vector $\ph$ from this domain. It is not
hard to see, using the Hardy and the diamagnetic inequalities,
that $\EE^{A,V}$ is bounded below and that every minimizing
sequence is bounded in $H^1_A(\R^3)$, see Lemma~\ref{D-cont}. We
set
\begin{equation}
    C^{A,V}(\lambda):=\inf\big\{\EE^{A,V}(\ph)\big| \ph\in H^1_A(\R^3),\ \|\ph\|^2=\lambda\big\} \label{minimum}
\end{equation}
where $\lambda>0$. As a preparation for the proofs of the theorems
of this section we first establish a few general properties of the
Pekar functional \eqref{pekar} and its lower bounds
\eqref{minimum}. To this end, and for use throughout the paper, we introduce the following notation:
\begin{equation*}
    V_{\ph}(x) := \int\frac{|\ph(y)|^2}{|x-y|}\,dy,\qquad
    D(\rho) :=  \int\frac{\rho(x)\rho(y)}{|x-y|}\,dxdy,
\end{equation*}
where usually $\rho=\rho_\ph:=|\ph|^2$.


\begin{lemma}\label{general-AV}
Under the above minimal assumptions on $V,A$, the following is true:
\begin{itemize}
\item[(i)] If $\EE^{A,V}(\ph_n)\to C^{A,V}(1)$ and $\ph_n\to\ph$ as $n\to\infty$, then  $\EE^{A,V}(\ph)=C^{A,V}(1)$ and  $\ph_n\to\ph$ in $H^1_A(\R^3)$.
\item[(ii)] If $\EE^{A,V}(\ph)=C^{A,V}(1)$, then $\ph$ is an eigenvector of $D_A^2+V- 2V_{\ph}$ associated with the lowest eigenvalue
of this operator, which is $C^{A,V}(1)-D(\rho_{\ph})$.
\item[(iii)] The map $\lambda\mapsto C^{A,V}(\lambda)$ is continuous.
\item[(iv)] If $\liminf_{n\to\infty} D(\rho_{\ph_n})>0$ for every (normalized) minimizing sequence of $\EE^{A,V}$, then for all $\lambda\in (0,1)$,
$$C^{A,V}(1) <  C^{A,V}(\lambda) +  C^{A,V} (1-\lambda).$$
\end{itemize}
\end{lemma}

\begin{proof}
(i) Since $(\ph_n)$ is bounded in $H^1_A(\R^3)$ and $\ph_n\to\ph$
we see that $\ph_n\rightharpoonup\ph$ in $H^1_A(\R^3)$, and hence
that $\EE^{A,V}(\ph)\leq \liminf_{n\to\infty}\EE^{A,V}(\ph_n)$, by
Lemma~\ref{D-cont}, (ii). It follows that
$\EE^{A,V}(\ph)=C^{A,V}(1)=\lim_{n\to\infty}\EE^{A,V}(\ph_n)$ and,
using Lemma~\ref{D-cont} again, that $\|D_A\ph_n\|^2\to
\|D_A\ph\|^2$. This proves (i).

(ii) We claim that
\begin{equation}\label{p-linearized}
      \EE^{A,V}(\psi) \leq \sprod{\psi}{(D_A^2+V- 2V_{\ph})\psi} + D(\rho_{\ph})
\end{equation}
for any given $\psi\in H^1_A(\R^3)$. This follows from $0\leq
D(\rho_\ph-\rho_\psi)=D(\rho_\ph)+D(\rho_\psi) -
2\sprod{\psi}{V_{\ph}\psi}$. If $\ph$ is a minimizer of
$\EE^{A,V}$, then it follows from \eqref{p-linearized} that for
every normalized $\psi\in H^1_A(\R^3)$,
$$
     C^{A,V}(1) \leq \sprod{\psi}{(D_A^2+V - 2V_{\ph})\psi} + D(\rho_{\ph})
$$
with equality if $\psi=\ph$. This proves part (ii).

(iii) Clearly for all $\lambda>0$,
\begin{equation}
C^{A,V}(\lambda) =\lambda\cdot\inf\big\{\norm{D_A\ph}^2+\sprod{\ph}{V\ph} -\lambda D(\rho_{\ph})\big| \|\ph\|=1\big\}. \label{cavlambda}
\end{equation}
We see that $g(\lambda) =C^{A,V}(\lambda)/ \lambda$ is the infimum of linear functions of $\lambda$. It follows that $g$ is concave and hence continuous.

(iv) It suffices to show that
\begin{equation}
C^{A,V}(\lambda)>\lambda C^{A,V}(1)\quad\text{for all}\quad\lambda\in (0,1). \label{xy}
\end{equation}
Then $C^{A,V}(1-\lambda)>(1-\lambda )C^{A,V}(1)$ and the asserted
inequality follows. Since, by \eqref{cavlambda}, $C^{A,V}(\lambda
)\ge\lambda C^{A,V} (1)$, it remains to exclude equality. Again by
\eqref{cavlambda}, the equality $C^{A,V}(\lambda)=\lambda
C^{A,V}(1)$ would imply the existence of a normalized sequence
$(\ph_n)$ with $\|D_A\ph_n\|^2+\sprod{\ph_n}{V\ph_n} - \lambda
D(\rho_{\ph_n}) \to C^{A,V}(1)$. A fortiori, this sequence would
be minimizing for $\EE^{A,V}$ and $D(\rho_{\ph_n})\to 0$, in
contradiction with the assumption.
\end{proof}


\begin{lemma}\label{Alinear}
If $A$ is \emph{linear} with $B=\curl A$, then
\begin{itemize}
\item[(i)] $C^{0,0}(1)\le C^{A,0}(1)\le C^{0,0}(1)+|B|$, and $C^{0,0}(1)<0$.
\item[(ii)] If $(\ph_n)$ is a minimizing sequence for $\EE^{A,0}$ then $\liminf_{n\to\infty} D(\rho_{\ph_n})>0$.
\end{itemize}
\end{lemma}


\begin{proof}
The inequality $C^{0,0}(1)\leq C^{A,0}(1)$ follows from the
diamagnetic inequality, and $C^{0,0}(1)<0$ follows from a simple variational argument. By combining \eqref{p-linearized} with the
enhanced binding inequality of Lieb \cite{AHS78}, we
conclude that, for $\ph\in H^1(\R^3)$ with $\|\ph\|=1$,
\begin{eqnarray*}
    C^{A,0}(1) &\leq & \inf\sigma(D_{A}^2 - 2V_{\ph}) + D(\rho_\ph)\\
    & \leq &\inf\sigma(-\Delta - 2V_{\ph}) + D(\rho_\ph) +|B| \\
    & \leq &\sprod{\ph}{(-\Delta- 2V_{\ph})\ph} + D(\rho_{\ph}) +|B|\\
   &=& \EE^{0,0}(\ph) +|B|.
\end{eqnarray*}
To prove (ii), suppose that $D(\rho_{\ph_n})\to 0$ as $n\to\infty$
for some minimizing sequence $(\ph_n)$ of $ \EE^{A,0}$. Then
\begin{equation}\label{cavgrb}
C^{A,0}(1) = \lim_{n\to\infty}\EE^{A,0}(\ph_n)=\lim_{n\to\infty}\|D_A\ph_n\|^2 \ge |B|,
\end{equation}
which is in contradiction with the fact that $C^{A,0}(1)\leq C^{0,0}(1)+|B|<|B|$, by (i).
\end{proof}


\begin{theorem}\label{linA-thm}
Suppose that $A$ is linear. Then there exists a $\ph\in H^1_A(\R^3)$ with $\int |\ph |^2~dx = 1$ such that
$$
\EE^{A,0} (\ph ) = C^{A,0}(1),
$$
and every minimizing sequence for $\EE^{A,0}$ has a subsequence that converges to a minimizer
after suitable translations and phase shifts.
\end{theorem}

\noindent
\emph{Remark.} The Pekar functional $\EE^{A,0}$ with a linear vector potential $A$ is invariant under magnetic translations 
$\psi \mapsto \psi_v$, $v\in\R^3$, where   
\begin{equation}\label{mag-trans}
      \psi_v(x):=e^{-i\chi(x)}\psi(x-v),\qquad \chi(x):=A(v)\cdot x,\quad v\in\R^3.
\end{equation}
This means that minimizing sequences will in general not be relatively compact. 
By the concentration compactness principle every minimizing sequence has a subsequence 
that becomes relatively compact upon suitable 
translations of the type \eqref{mag-trans}. 


\begin{proof}
Let $(\ph_n)$ be a minimizing sequence for $\EE^{A,0}$ and let
$(\ph_{n_k})$ be the subsequence given by Lemma~\ref{Lions}. We
shall exclude vanishing and dichotomy in order to conclude
compactness of the sequence of suitably shifted functions. In the
following we use $\rho_n$ as a short hand for $\rho_{\ph_n}$.

\emph{Vanishing does not occur.} We show that vanishing implies
$D(\rho_{n_k})\to 0$ as $k\to\infty$, which contradicts
Lemma~\ref{Alinear}~(ii). To this end we use that
$D(\rho_\ph)=\int V_\ph\rho_\ph\,dx \leq \|V_\ph\|_{\infty}$ where
$\ph\in L^2(\R^3)$ is normalized.  For every $R>0$, by the H\"older and the magnetic
Hardy inequalities,
\begin{align*}
    |V_{\ph_{n_k}}(x)| &\leq \int_{B_R(x)} \frac{|\ph_{n_k}(y)|^2}{|x-y|}dy  + \frac{1}{R}\\
    & \leq 2 \|D_A \ph_{n_k} \| \left(\int_{B_R(x)} |\ph_{n_k}(y)|^2\, dy\right)^{1/2} + \frac{1}{R}.
\end{align*}
Since $\sup_{k}\|D_A \ph_{n_k} \|<\infty$, vanishing implies $\|V_{\ph_{n_k}}\|_{\infty} \to 0$ and $D(\rho_{n_k})\to 0$ as $k\to\infty$.

\emph{Dichotomy does not occur.} Suppose dichotomy holds, that is,
there exists some $\lambda\in (0,1)$, such that for every $\eps>0$
there exists $k_0\in \N$ and bounded sequences
  $(\ph_k^{(1)})$, $(\ph_k^{(2)})$ in $H^1_{A}(\R^3)$ having the properties (a)--(d) from Lemma~\ref{Lions}.
Then, from (a), (c) and the continuity of $\ph\mapsto
D(\rho_\ph)$, Lemma~\ref{D-cont}, we see that for $k\geq k_0$
\begin{eqnarray*}
    \lefteqn{\big|D(\rho_{n_k}) - D(\rho_{k}^{(1)}) - D(\rho_{k}^{(2)}) \big|} \\
    & \leq & \big|D(\rho_{n_k}) - D(|\ph_k^{(1)}+\ph_k^{(2)}|^2)\big| + \big|D(|\ph_k^{(1)}+\ph_k^{(2)}|^2)- D(\rho_{k}^{(1)}) - D(\rho_{k}^{(2)})\big| \\
    & = & \delta(\eps) + o(1),\qquad (k\to\infty),
\end{eqnarray*}
where $\delta(\eps)=o(1)$ as $\eps\to 0$. It follows that, using Lemma~\ref{general-AV} (iii) and Lemma~\ref{Lions}~(d),
\begin{eqnarray*}
& & C^{A,0}(1) \\
&=&\lim_{k\to \infty} \EE^{A,0} (\ph_{n_k} ) \\
&\geq &\liminf_{k\to\infty}\Big[\EE^{A,0} (\ph_{n_k} )-\EE^{A,0} (\ph_{k}^{(1)} ) -\EE^{A,0} (\ph_{k}^{(2)})\Big]+ C^{A,0}(\lambda )+C^{A,0} (1-\lambda ) +o(1)\\
&\geq &\liminf_{k\to\infty}\int_{\R^3}|D_A\ph _{n_k}|^2-|D_A\ph _{k}^{(1)}|^2-|D_A\ph _{k}^{(2)}|^2 \,dx
+C^{A,0}(\lambda )+C^{A,0}(1-\lambda) + o(1)\\
 &\geq & C^{A,0}(\lambda )+C^{A,0}(1-\lambda) + o(1),\qquad (\eps\to 0).
\end{eqnarray*}
This proves that $C^{A,0}(1) \geq C^{A,0}(\lambda )+C^{A,0}(1-\lambda)$ for some $\lambda \in (0,1)$, which contradicts Lemma~\ref{general-AV}~(iv).

\emph{Compactness.} Since vanishing and dichotomy have been
excluded, the subsequence $ (\ph_{n_k})$ must have the compactness
property of Lemma~\ref{Lions}. Let
$\chi_{k}(x):=A(y_k)\cdot x$ with $y_k\in\R^3$
given by this lemma, and let
$u_{n_k}(x)=e^{i\chi_k(x)}\ph_{n_k}(x+y_k)$. Then, for every
$\eps>0$ there exists $R>0$ such that
\begin{equation}\label{compact}
    \int_{B_R(0)}|u_{n_k}|^2 dx \geq 1-\eps\qquad \text{for all}\ k.
\end{equation}
The phase $\chi_k$ has been chosen in such a way that
$A(x)+\nabla\chi_k(x) = A(x+y_k)$, which implies that $\|D_A
u_{n_k}\| = \|D_A \ph_{n_k}\|$. It follows that
$\EE^{A,0}(u_{n_k})=\EE^{A,0}(\ph_{n_k})$ and that $(u_{n_k})$ is
bounded in $H_A^1(\R^3)$. Hence there exists a $u\in H_A^1(\R^3)$
and a subsequence of  $(u_{n_k})$, denoted by  $(u_{n_k})$ as
well, such that
\begin{equation}\label{A-weak}
     u_{n_k} \rightharpoonup u,\quad \text{in}\ H_A^1(\R^3),
\end{equation}
and therefore $u_{n_k} \rightharpoonup u$ in $L^2(\R^3)$. We claim
that $\|u\|=1$ and hence that $u_{n_k} \rightarrow u$ in
$L^2(\R^3)$. Indeed, since $A$ is locally bounded, \eqref{A-weak}
implies that $u_{n_k} \rightarrow u$ locally in $L^2(\R^3)$, and
by \eqref{compact} we conclude that
$$
    1\geq \|u\|^2 \geq \int_{B_R(0)}|u|^2\, dx = \lim_{k\to\infty}  \int_{B_R(0)}|u_{n_k}|^2\, dx \geq 1-\eps
$$
for every $\eps>0$. The theorem now follows from
Lemma~\ref{general-AV}~(i).
\end{proof}


We say $A$ is \emph{asymptotically linear} if there exists a
linear vector potential $A_{\infty}$ such that
$$|A(x)-A_{\infty}(x)|\to 0,\qquad \text{as}\ |x|\to\infty.$$
In addition we shall assume that $A\in L^3_{loc}(\R^3)$ whenever
$A$ is asymptotically linear. This technical assumption ensures,
e.~g. that $H^1_A(\R^3) = H^1_{A_{\infty}}(\R^3)$ and that the
norms of these spaces are equivalent (see Lemma \ref{asylin}).

To ensure relative compactness of minimizing sequences we shall impose one of the following \emph{trapping assumptions}:
\begin{itemize}
\item[(T1)] $V(-\Delta +1)^{-1}$ is compact and
$$C^{A,V}(1)< C^{A,0}(1).$$
\item[(T2)] $V(-\Delta +1)^{-1}$ is compact, $A$ is asymptotically linear and
$$C^{A,V}(1)< C^{A_{\infty},0}(1).$$
\end{itemize}
Further below we shall give examples of potentials that satisfy either (T1) or (T2).


\begin{theorem}\label{trap-thm}
Suppose that one of the trapping assumptions (T1) or (T2) is
satisfied. Then every minimizing sequence of $\EE^{A,V}$ has a
convergent subsequence, the limit being a minimizer.
\end{theorem}


\noindent
\emph{Remark.} If $V(-\Delta + 1)^{-1}$ is compact and $A$ is
asymptotically linear, then the inequality
$C^{A,V}(1)<C^{A_{\infty}, 0}(1)$ is not only sufficient, but also
necessary for the conclusion of Theorem~\ref{trap-thm} to hold.


\begin{proof}
Let $(\ph_n)$ be a minimizing sequence for $\EE^{A,V}$. After
passing to a subsequence we may assume that $\ph_n
\rightharpoonup\psi$ in $H^1_A(\R^3)$. We claim that $\psi=0$ is
in contradiction with (T1) and (T2). Indeed, if
$\ph_n\rightharpoonup 0$ then $\sprod{\ph_n}{V\ph_n}\to 0$, by
Lemma~\ref{local-VA}, which implies that $C^{A,V}(1)\ge C^{A,0}(1)$
in contradiction with (T1). If $A$ is asymptotically linear, then
$D_A\ph_n = D_{A_{\infty}}\ph_n + (A-A_{\infty})\ph_n$ where
$(A-A_{\infty})\ph_n\to 0$ by Lemma~\ref{asylin}. It follows that
$$C^{A,V}(1) = \lim_{n\to\infty}\EE^{A,V}(\ph_n)=\lim_{n\to\infty}\EE^{A_{\infty},0}(\ph_n)\ge C^{A_{\infty},0}(1).$$
This is in contradiction with (T2).

Using that the weak limit of a minimizing sequence cannot vanish, we conclude, from  Lemma~\ref{D-cont}~(iii), that
$$
\liminf_{n\to\infty}D(\rho_{\ph_n})>0
$$
for every minimizing sequence $(\ph_n)$. It follows that 
$\lambda\mapsto C^{A,V}(\lambda )$ is subadditive in the
sense of Lemma~\ref{general-AV}. We now use this to show that 
a weakly convergent minimizing sequence  $(\ph_n)$ is in fact strongly convergent. To this 
end suppose that $\ph_n \rightharpoonup \psi$ where $\lambda:=\|\psi\|^2\in (0,1)$
and consider the decomposition $\ph_n = \psi+(\ph_n-\psi) =: \psi +
\beta_n$. Clearly, $\beta_n\rightharpoonup 0$ in $H^1_A(\R^3)$ and
$\|\beta_n\|^2\to 1-\lambda$. We claim that
\begin{equation}\label{enerdecomp}
\EE^{A,V}(\ph_n) =\EE^{A,V}(\psi)+\EE^{A,V}(\beta_n)+o(1),\qquad (n\to\infty).
\end{equation}
The kinetic and potential energy $\|D_A\ph_n\|^2
+\sprod{\ph_n}{V\ph_n}$ decompose as desired, which is a direct
consequence of the weak convergence $\beta_n\rightharpoonup 0$ in
$H^1_A(\R^3)$ and the compactness of $V(-\Delta + 1)^{-1}$. It is
not hard to see, using $\beta_n\to 0$ locally in $L^2(\R^3)$, that
\begin{equation*}
D(\rho_{\psi+\beta_n}) = D(\rho_{\psi}) + D(\rho_{\beta_n}) + o(1),\qquad (n\to\infty).
\end{equation*}
From \eqref{enerdecomp} we see that
\begin{eqnarray*}
\EE^{A,V}(\ph_n) &\ge& C^{A,V}(\lambda ) +C^{A,V}(\|\beta_n\|^2) + o(1)\\
&=&  C^{A,V}(\lambda ) +C^{A,V}(1-\lambda) + o(1)
\end{eqnarray*}
for $n\to\infty$, by the continuity of $C^{A,V}$ (Lemma~\ref{general-AV}~(iii)). Thus $C^{A,V}(1)\ge C^{A,V}(\lambda) + C^{A,V}(1-\lambda)$
which contradicts the subadditivity of $C^{A,V}$, i.e. Lemma \ref{general-AV}~(iv).

Since we have shown that $\|\psi\|<1$ is impossible, we conclude that $\|\psi\| =1$ and $\ph_n\to\psi$ in $L^2(\R^3)$.
The theorem now follows from Lemma~\ref{general-AV}~(i).
\end{proof}


\noindent\textbf{Examples:}
\begin{itemize}
\item[1)] Suppose $A$ is any $C^1$-vector potential for which $\EE^{A,0}$ has a minimizer $\ph$, see Theorems~\ref{linA-thm} and \ref{trap-thm}.
Then the Euler-Lagrange equation satisfied by $\ph$ is a Schrödinger equation and hence $\ph$ cannot vanish a.e. on a non-trivial open set, see \cite{Kalf82}. It follows that $\int V|\ph|^2~dx <0$ for every potential
$V\le 0$ with the property that $V<0$ on some non-empty open set. If,
moreover, $V(-\Delta + 1)^{-1}$ is compact, then (T1) is
satisfied.
\item[2)] We choose $V=0$ and we define the vector potential $A$ by $A=A_R$ where
$$
   A_R(x) = \begin{cases}0, &|x|<R\\ A_{\infty}(x), &|x|\geq R\end{cases}
$$
and $A_{\infty}(x)=(-Bx_2,0,0)$. We claim that
$C^{A,0}(1)<C^{A_\infty,0}(1)$ for $B\geq 4$ and $R$ sufficiently
large. Indeed, by Lemma~\ref{D-cont}, $\EE^{A_{\infty},0}(\ph) =
\|D_{A_{\infty}}\ph\|^2-D(\rho_\ph)\geq B - 2\|\ph\|^3
\|D_{A_{\infty}}\ph\|\geq 0$, while $C^{A_R,0}(1)\to C^{0,0}(1)<0$
as $R\to\infty$.
\end{itemize}


The following corollary summarizes the conclusions of Example 1) above and Theorem~\ref{trap-thm}.

\begin{corollary}\label{suffcond}
Suppose that $V(-\Delta+1)^{-1}$ is compact, $V\le 0$, and $V<0$ on some non-empty open set.
Then $\EE^{A,V}$ has a minimizer, provided $\EE^{A,0}$ has a minimizer and $A$ belongs to $C^1$. 
In particular $\EE^{A,V}$ has a minimizer for every linear vector potential $A$.
\end{corollary}


\section{Binding of Polarons}
Let $V$ and $A$ satisfy the minimal assumption introduced in the previous section.
The magnetic Pekar-Tomasevich functional $\EE^{A,V}_{U}: H^1_{(A,A)}(\R^6)\to \R$ is defined by
\begin{align*}
\EE^{A,V}_{U} (\psi) :=& \sum_{k=1}^2\int \left(|D_{A,x_k}\psi(x_1,x_2)|^2+V(x_k)|\psi(x_1,x_2)|^2\right) dx_1dx_2 \\
&+ U\int \frac{|\psi(x_1,x_2)|^2}{|x_1-x_2|}dx_1dx_2 - \int\frac{\rho (x_1)\rho (x_2)}{|x_1-x_2|}dx_1dx_2,
\end{align*}
where
$$
\rho(x):= \int (|\psi(x,y)|^2+|\psi(y,x)|^2 )dy
$$
denotes the density. The minimal energy of $\EE^{A,V}_{U}$ is defined by
$$
C^{A,V}_{U} = \inf\left\{ \left. \EE^{A,V}_{U} (\psi)\right| \psi\in H^1_{(A,A)}(\R^6), \|\psi\| =1 \right\} .
$$

\begin{theorem}\label{NL-binding}
Suppose that $\EE^{A,V}$ possesses a minimizer $\ph_0$; see Theorem~\ref{linA-thm}, 
Theorem~\ref{trap-thm}, and Corollary~\ref{suffcond}. 
Then there exists $U_A>2$ such that for $2<U<U_A$ we have
$$
  C^{A,V}_{U} <2C^{A,V}(1).
$$
\end{theorem}

\begin{proof}
Since $C^{A,V}_{U}$ is continuous with respect to $U$ it suffices
to prove that $C^{A,V}_{U}<2C^{A,V}(1)$ for $U=2$. By a straightforward
computation
$$
 \EE^{A,V}_{U=2}(\ph_0\otimes\ph_0) = 2\EE^{A,V}(\ph_0) = 2C^{A,V}(1),
$$
and it remains to prove that $\ph_0\otimes\ph_0$ is not a
minimizer of $\EE^{A,V}_{U=2}$. To this end, suppose $\ph_0\otimes\ph_0$
were a minimizer of $\EE^{A,V}_{U=2}$. Then it would have to solve the Euler equation of the functional, which implies that
\begin{equation}\label{PT-euler}
     \form{\eta \otimes \eta}{\sum_{k=1}^2(D_{A,x_k}^2+V(x_k)-4V_{\ph_0}(x_k)) +2|x_1-x_2|^{-1}-E}{\ph_0\otimes\ph_0}=0
\end{equation}
for some $E$ and all $\eta \in H^1_A(\R^3)$. 
We claim that \eqref{PT-euler} cannot be true for all $\eta$.
Since $\ph_0$ minimizes $\EE^{A,V}$, we know from Lemma~\ref{general-AV}~(ii), that
$(D_{A}^2+V-2V_{\ph_0})\ph_0 = \lambda\ph_0$
for some $\lambda\in\R$. Hence equation \eqref{PT-euler} reduces to
\begin{equation} \label{eqn:q}
    \form{\eta \otimes \eta}{2\lambda-E -2\sum_{k=1}^2 V_{\ph_0}(x_k)+2|x_1-x_2|^{-1}}{\ph_0\otimes\ph_0}=0
\end{equation}
for all $\eta\in H^1_A(\R^3)$. Since $V_{\ph_0}$ is bounded while 
$|x_1-x_2|^{-1}$ is positive and unbounded, we can choose $r>0$ so that for all $z \in \R^3$ and all
$x_1,x_2\in B_r(z)$,
\begin{equation}\label{eqn:ge1}
g(x_1,x_2) := 2\lambda-E -2\sum_{k=1}^2 V_{\ph_0}(x_k) +2|x_1-x_2|^{-1} \geq 1.
\end{equation}
Let $\chi_{(r,z)} \in C_0^\infty(\R^3;[0,1])$ with $\chi_{(r,z)}(x) =1$ for $x \in B_{r/2}(z)$ and $\chi_{(r,z)}(x) = 0$ for $x \not\in B_r(z)$.
In view of \eqref{eqn:ge1} the choice $\eta = \chi_{(r,z)}\ph_0$ in \eqref{eqn:q} leads to 
$$
0 = \form{\chi_{(r,z)} \ph_0 \otimes \chi_{(r,z)} \ph_0}{g}{\ph_0 \otimes \ph_0} \geq \left( \int_{B_{r/2}(z)} |\ph_0(x)|^2 \,dx \right)^2,
$$
for all $z\in\R^3$. It follows that $\ph_0=0$ in contradiction with $\norm{\ph_0}=1$.
\end{proof}


\appendix
\section{Appendix}

The following is a variant of the Lions' concentration compactness
principle, Lemma III.1, in \cite{Lions1984}, the only difference
being that $D=-i\nabla$ is replaced by $D_A$ in our version. This
does not affect the proof.

\begin{lemma}[Concentration Compactness Lemma]\label{Lions}
Suppose that $A:\R^3\to\R^3$ is real-valued and in $L^2_{loc}(\R^3)$.
Let $(\ph_n)_{n\in\N}$ be a bounded sequence in $H^1_A(\R^3)$, let
$\rho_n=|\ph_n|^2$ and suppose
$$
     \int \rho_n(x)dx =1\qquad \text{for all}\ n\in\N.
$$
Then there exists a subsequence $(\ph_{n_k})$ which has one of the
following three properties:
\begin{enumerate}
\item \emph{Compactness:} There exists a sequence $(y_k)_{k \geq 0} \subset \R^3$
  such that for all $\eps >0$ there is $R>0$ with
$$
    \int_{B_R(y_k)} \rho_{n_k}(x)dx
    \geq 1-\eps \quad \text{ for all } k \geq 0.
$$
\item \emph{Vanishing:} For all $R>0:$
$$
    \lim_{k\to\infty} \left(\sup_{y\in\R^3} \int_{B_R(y)} \rho_{n_k}(x)dx\right) = 0.
$$
\item \emph{Dichotomy:} There exists $\lambda\in(0,1)$ such that for every
  $\eps>0$ there exists $k_0\in\N$ and bounded sequences
  $(\ph_k^{(1)})$, $(\ph_k^{(2)})$ in $H^1_{A}(\R^3)$ satisfying,
\begin{itemize}
\item[(a)] $\displaystyle\|\ph_{n_k} - (\ph_k^{(1)}+\ph_k^{(2)})\| = \delta (\eps ),\quad k\geq k_0,$
\item[(b)] $\displaystyle|\|\ph_k^{(1)}\|^2-\lambda |\le\eps , \quad|\|\ph_k^{(2)}\|^2-(1-\lambda )|\le\eps,\quad k\geq k_0,$
\item[(c)] $\displaystyle\dist (\supp (\ph_k^{(1)}), \supp (\ph _k^{(2)} )) \rightarrow\infty\quad (k\to\infty ),$
\item[(d)] $\displaystyle\liminf_{k\to\infty}\int \big(|D_A\ph_{n_k} (x)|^2 - |D_A\ph_k^{(1)} (x)|^2 - |D_A\ph_k^{(2)} (x)|^2\big) dx\ge 0,$
\end{itemize}
where $\delta(\eps)\to 0$ as $\eps\to 0$ in property (a).
\end{enumerate}
\end{lemma}


\begin{lemma}\label{D-cont}
Under our minimal assumptions on $A,V$ the following is true:
\begin{itemize}
\item[(i)] $D(\rho_\ph) \leq 2\|\ph\|^3\|D_A\ph\|$ for all $\ph\in
H_A^1(\R^3)$.
\item[(ii)] On bounded subsets of $H_A^1(\R^3)$ the maps
$\ph\mapsto \sprod{\ph}{V\ph}$ and $\ph\mapsto D(\rho_\ph)$ are
continuous w.r.t. the norm of $L^2(\R^3)$.
\item[(iii)]  In  $H^1_A(\R^3)$ the map $\ph\mapsto D(\rho_\ph)$
is weakly lower semi-continuous.
\item[(iv)] For every $\eps \in (0,1)$ there
exists $C_\eps$ such that for all $\ph\in H_A^1(\R^3)$
$$
   \|D_A\ph\|^2 \leq \frac{1}{1-\eps}\EE^{A,V}(\ph) +
   C_\eps\big(\|\ph\|^2+\|\ph\|^6\big).
$$
\end{itemize}
\end{lemma}

\begin{proof}
(i) We have
$D(\rho_\ph)=\int\rho_{\ph}(x)V_\ph(x)dx \leq \|\rho_\ph\|_1
\|V_\ph\|_\infty$, where
$$
   \|V_\ph\|_\infty \leq \|\ph\|\left(\int\frac{|\ph(y)|^2}{|x-y|^2}\,dy\right)^{1/2}
   \leq 2\|\ph\| \|\nabla|\ph|\|,
$$
by the H\"older and the Hardy inequalities. (i) now follows from
the diamagnetic inequality $|\nabla|\ph||\leq|D_A\ph|$.

(ii) The continuity of  $\ph\mapsto D(\rho_\ph)$
follows from
\begin{align*}
     | D(\rho_{\ph}) -  D(\rho_{\psi}) |&= \left| \int\big(\rho_{\ph}(x)-\rho_{\psi}(x)\big)\big(V_\ph(x)+V_\psi(x)\big) \,dx \right|\\
      &\leq \|\rho_{\ph}-\rho_{\psi}\|_1\big(\|V_{\ph}\|_{\infty}+\|V_{\psi}\|_{\infty}\big)
\end{align*}
where $\|\rho_{\ph}-\rho_{\psi}\|_1 \leq
\|\ph-\psi\|(\|\ph\|+\|\psi\|)$ and $\|V_{\ph}\|_{\infty}\leq 2
\|D_A\ph\| \|\ph\|$, by (i). We now turn to the map $\ph\mapsto
\sprod{\ph}{V\ph}$. The assumption $V\ll -\Delta$ is equivalent to
$|V|\ll -\Delta$ which implies that $|V|\leq
\eps(-\Delta)+C_{\eps}$ for all $\eps>0$. From here the continuity
of $\ph\mapsto \sprod{\ph}{V\ph}$ is easily established.

(iii) Let $\chi\in C_0^{\infty}(\R^3;[0,1])$ with $\chi(x)=1$ for
$|x|\leq 1$ and let $\chi_R(x):=\chi(x/R)$. The weak convergence
$\ph_n\rightharpoonup\ph$ in $H^1_A(\R^3)$ implies the norm
convergence $\chi_R\ph_n\to\chi_R\ph$ in $L^2(\R^3)$. This can be
seen from Lemma~\ref{local-VA} with the choice $V=\chi_R^2$. Since
the sequence $(\chi_R\ph_n)$ is bounded in $H^1_A(\R^3)$, it
follows from (ii) that $\liminf_{n\to\infty}D(\rho_{\ph_n})\geq
\liminf_{n\to\infty}D(\chi_R^2\rho_{\ph_n})=D(\chi_R^2\rho_\ph)$ for all $R>0$ and
the desired inequality is obtained using monotone convergence.

(iv) The assumption $V\ll -\Delta$ and the diamagnetic inequality imply that
$\eps D_A^2+V$ is bounded below for every $\eps>0$. With the help of (i) the inequality in (iv) now easily follows.
\end{proof}

\begin{lemma}\label{asylin}
\begin{itemize}
\item[(i)] If $A_1,A_2$ belong to $L^3_{\rm loc}(\R^3;\R^3)$ and $A_1-A_2$ is uniformly bounded in the complement of some compact set, then
$H^1_{A_1}(\R^3) = H^1_{A_2}(\R^3)$ and the corresponding norms $\|\cdot \|_{A_1}$ and $\|\cdot \|_{A_2}$ are equivalent.
\item[(ii)] If $A$ is asymptotically linear, then the linear map $H^1_A(\R^3)\to L^2(\R^3;\C^3)$, $\ph\mapsto (A-A_{\infty})\ph$ is compact. 
\end{itemize}
\end{lemma}

\noindent
\emph{Remark.} Further embedding results similar to Lemma~\ref{asylin} can be found in \cite{EstLions1989}.

\begin{proof}
(i) Suppose that $|A_1-A_2|\leq C$ in the complement of the compact set $K\subset\R^3$. Then, for all $\ph\in C_0^{\infty}(\R^3)$,
$\|D_{A_2}\ph\|\leq \|D_{A_1}\ph\|+ \|(A_1-A_2)\ph\|$ and
\begin{eqnarray*}
\|(A_1-A_2)\ph\|^2 &\le& \int_{K}|A_1-A_2|^2|\ph |^2dx + C^2\|\ph\|^2\\
&\le& \left(\int_{K}|A_1-A_2|^3dx\right)^{2/3}\|\ph\|^2_6+ C^2\|\ph\|^2.
\end{eqnarray*}
Since $\|\ph\|_6 \le \const\|D_{A_1}\ph\|$ by the Sobolev and the
diamagnetic inequalities, it follows that $\|D_{A_2}\ph\|\leq
\const\|\ph\|_{A_1}$ for all $\ph\in C_0^{\infty}(\R^3)$. This
extends to all $\ph\in H^1_{A_1}(\R^3)$ and then proves the lemma
since the roles of $A_1$ and $A_2$ are interchangeable.

(ii) The boundedness of the map has been established in the proof of (i). To prove the compactness, let $(\ph_n)$ be a bounded sequence in $ H^1_A(\R^3)$.
After passing to a subsequence we may assume that $\ph_n\rightharpoonup \ph$ in $H^1_A(\R^3)$. By the Sobolev inequality, the sequence $(|\ph_n-\ph |^2 )$
is bounded in $L^3(\R^3)$, which is a reflexive Banach space. Hence we may assume that  $|\ph_n-\ph|^2\rightharpoonup \psi$ in $L^3(\R^3)$ by passing
to a subsequence once more.
We claim that $\psi =0$. Indeed, from  $\ph_n\rightharpoonup \ph$ in $H^1_A(\R^3)$ it follows that $\int\chi |\ph_n-\ph|^2\,dx\to 0$ for $\chi\in C_0^{\infty}(\R^3)$,
as explained in the proof of Lemma~\ref{D-cont}~(iii). On the other hand, $\int\chi |\ph_n-\ph|^2\,dx\to\int\chi\psi\,dx$ because
$C_0^{\infty}(\R^3)\subset L^{3/2}(\R^3)$, which is the dual of $L^3(\R^3)$.
Thus $\int\chi\psi\,dx =0$ for all $\chi\in C_0^{\infty}(\R^3)$, which implies $\psi =0$. Hence $|\ph_n-\ph|^2\rightharpoonup 0$ in $L^3(\R^3)$
and it is easy to see that $(A-A_{\infty})(\ph_n-\ph)\to 0$ in $L^2(\R^3;\C^3)$ using that $|A-A_{\infty}|\leq \eps$ on the complement of some ball
$B_R$ and that $\chi_{B_R}|A-A_{\infty}|^2$ belongs to $L^{3/2}(\R^3)$, the dual of $L^3(\R^3)$.
\end{proof}


\begin{lemma}\label{local-VA}
In addition to the minimal assumptions on $A,V$, suppose that
$V(-\Delta+1)^{-1}$ is compact. Then the map $\ph\mapsto \sprod{\ph}{V\ph}$ is weakly continuous in $H_A^1(\R^3)$.
\end{lemma}

\begin{proof}
The compactness of $V(-\Delta+1)^{-1}$ implies that
$V(D_A^2+1)^{-1}$ is compact \cite{AHS78}.
By interpolation it follows that $(D_A^2+1)^{-1/2}V(D_A^2+1)^{-1/2}$ is compact, which implies that
 $\ph\mapsto \sprod{\ph}{V\ph}$ is weakly continuous.
\end{proof}
 

\end{document}